\documentclass[conference]{IEEEtran}
\IEEEoverridecommandlockouts

\usepackage{hyperref}
\usepackage{cite}
\usepackage{amsmath,amssymb,amsfonts}
\usepackage{graphicx}
\usepackage{textcomp}
\usepackage{cuted}
\usepackage{flushend}
\def\BibTeX{{\rm B\kern-.05em{\sc i\kern-.025em b}\kern-.08em
    T\kern-.1667em\lower.7ex\hbox{E}\kern-.125emX}}

\newtheorem{theorem}{Theorem}[section]
\newtheorem{proof}{Proof}[section]

\begin{document}

\title{State Estimation Of a Quantum Cavity Driven by single-photon\\
\thanks{Email: a-daeichian@araku.ac.ir, a.daeichian@gmail.com; The Author acknowledges support from Arak University (Project No: 96/15582-1396/12/21).}
}

\author{\IEEEauthorblockN{Abolghasem Daeichian}
\IEEEauthorblockA{\textit{Department of Electrical Engineering, Faculty of Engineering},
\textit{Arak University},
Arak, 38156-8-8349, Iran \\
\textit{Institute of advanced technology}, \textit{Arak University}, Arak, 38156-8-8349, Iran}
}

\maketitle

\begin{abstract}
The cavity is a fundamental ingredient of quantum optical systems. This paper concerns the behavior of a quantum cavity driven by non-classical field in single-photon state. To this end, the number operator has been opted to reveal the number of photons inside the cavity. Then, the quantum filtering equations have been employed to derive a stochastic master equation for the cavity which is driven by single-photon and is observed by either Homodyne or photon-counting detector. Finally, the state of the cavity has been estimated by the derived equations, and the results have been compared with the conventional master equation.
\end{abstract}


%%%%%%%%%%%%%%%%%%%%%%%%%%%%%%%%%%%%%%%%%%%%%%%%%%%%%%%%%%%%%%%%%%%%%%%
\section{Introduction}
Nowadays, technologies based on quantum theory have been realized and implemented as well as progressing rapidly ~\cite{RN1339,RN1338}. Developing quantum control theory and practice is preliminary to create quantum systems \cite{RN332,DS2012}. Feedback control based on measurements is an effective method which is frequently employed. The Markovian feedback is denoted to control laws which are a simple function of measures, and the dynamical equation is given by the master equation (ME). However, in many cases, the feedback law depends on an estimation of the state, which is called Bayesian method \cite{RN1343}. The state estimation (filtering) is not only used in feedback control but also disclose the behavior of quantum systems which are not possible to find out directly without demolition \cite{RN1290}. 

Quantum filtering (unraveling) was firstly founded by Belavkin \cite{belavkin1989nondemolition}. It is known as stochastic master equation (SME) in quantum optic which represents the stochastic evolution of system conditioned on measurements. The filtering framework for a system driven by field in a vacuum, Gaussian, squeezed, or coherent states are developed \cite{RN25}. Also, filtering equation for a quantum system which is driven by field in non-classical states, particularly single-photon state, has been derived both in \cite{JMN2011,gough2012quantum} and by considering a non-Markovian approach in \cite{RN1307}. The filtering equation is also derived by assuming detector efficiency and utilizing multiple measurements \cite{RN1308}. 

The fields in non-classical states including single-photon, and the cavity are both fundamental components of photonic quantum systems \cite{RN1337}. A single-photon has been experimentally generated by different methods, such as exploiting the single-photon emission from natural and artificial atoms \cite{RN1340} or utilizing correlated photon pairs in nonlinear crystals \cite{RN1341} or entangled states \cite{RN1342}. Also, the behavior of a cavity has been under attention \cite{RN1157}.

This study concerns the dynamic of a cavity when it is driven by field in single-photon state. The behavior of a cavity can be investigated by the number of photons inside the cavity in time. Thus, the number operator has been selected as a system operator to estimate the number of photons inside the cavity by using filtering equations. The filtering equation has been derived for both Homodyne and photon-counting measurements. The derived SME has been simulated, and the results are compared with ME.

The paper has been folded as follows. The SLH representation of a cavity is presented in section \ref{Section_Model}. In Section \ref{Section_filtering}, quantum filtering of a system driven by field in single-photon state is reviewed from literature. The filtering equations of a cavity driven by a single-photon and observed by Homodyne or photon-counting detector are derived in section \ref{Section_filtering_Cavity}. The behavior of the cavity is simulated in section \ref{Section_Results}. Finally, the paper is concluded in section \ref{Section_Conclusion}.

%%%%%%%%%%%%%%%%%%%%%%%%%%%%%%%%%%%%%%%%%%%%%%%%%%%%%%%%%%%%%%%%%%%%%%%
\section{Cavity Model}
\label{Section_Model}
The state of a quantum system is denoted by a \emph{ket} $|\psi\rangle$ (typically a column complex vector in a Hilbert space) or density operator $\rho=|\psi\rangle\langle\psi|$ where the \emph{bra} $\langle\psi|$ is the Hermitian conjugate of the ket. Any measurable quantity (observable) of a quantum system is associated with a hermitian operator $X$ which acts on the state space and results to a real physical value \cite{S1994}.

Usually, unconditional evolution of density operator or observable (ME) is given for quantum systems. If the joint system state and field \{system$\otimes$field\} denoted by $\rho_{sf}=\rho\otimes\rho_{field}$, then the system state is given by tracing (averaging) over the field. 
A field may be observed by homodyne or photon-counting detector which gives a classical output signal $Y(t)$. SME for $\rho$ or $X$ \cite{gough2012quantum} or stochastic Schrödinger equation (SSE) for $|\psi\rangle$ present estimated state conditioned on measurements $Y(t)$ \cite{RN1157}. The ME is obtained by classical averaging over SME or SSE. 

An optical cavity (or resonator) is a fundamental component of optical systems which may be employed as a benchmark to assess control algorithms in the quantum regime. It is an arrangement of at least two mirrors that forms a standing wave. Two facing mirrors is the most popular structure of optical cavities.

\begin{figure}[!htb]
	\begin{center}
		\includegraphics[width=0.4\textwidth]{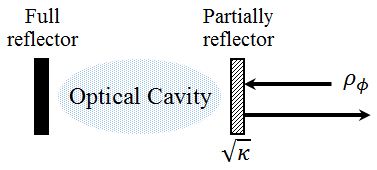}
		\caption{Schematic of a typical one-sided Cavity}
		\label{FigCavity}
	\end{center}
\end{figure}

Consider a \emph{one-sided} cavity which consists of a perfect reflecting mirror and another one with photon decay rate $\kappa$; see Fig.\ref{FigCavity}. The inter-cavity quantized field has creation and annihilation operators $a^\dagger$ and $a$, respectively, which adhere to the commutation relation $[a,a^\dagger]=aa^\dagger-a^\dagger a=I$ where $I$ is the identity and $\dagger$ means the Hermitian transpose. The matrix representation of the annihilation and creation operators are
\begin{eqnarray}\label{}
a=\left[a_{ij}\right]&=&\left[\begin{array}{cccc}
			0 & \sqrt{1} & 0        & \cdots \\
			0 & 0        & \sqrt{2} & \cdots \\
			0 & 0        & 0        & \cdots \\
   	   \vdots & \vdots   & \vdots	& \ddots\end{array}\right]\\
a^\dagger=\left[a_{ij}^\dagger\right]&=&\left[\begin{array}{cccc}
			0 & 0 		 & 0        & \cdots \\
 	 \sqrt{1} & 0        & 0 		& \cdots \\
			0 & \sqrt{2} & 0        & \cdots \\
	   \vdots & \vdots   & \vdots	& \ddots\end{array}\right].
\end{eqnarray}

The energy eigenstates are denoted as $|n\rangle$, $n=0,1,2,\cdots$ where $a|n\rangle=\sqrt{n}|n-1\rangle$, $a^\dagger|n\rangle=\sqrt{n+1}|n+1\rangle$,  and $a|0\rangle=0$. The Hamiltonian of the cavity is given by $H=\Delta a^\dagger a$ where $\Delta$ is the frequency detuning \cite{schleich2011quantum}. Overall, the cavity model in SLH-framework gives \cite{combes2017slh}:
\begin{equation}\label{Model_1S_Cavity}
G_c=\left(I,\sqrt{\kappa}a,\Delta a^\dagger a\right).
\end{equation}
where $S=I$, $L=\sqrt{\kappa}a$, and $H=\Delta a^\dagger a$ are scattering matrix, coupling operator, and Hamiltonian of the cavity, respectively.

%%%%%%%%%%%%%%%%%%%%%%%%%%%%%%%%%%%%%%%%%%%%%%%%%%%%%%%%%%%%%%%%%%%%%
\section{Quantum Filtering under single-photon}
\label{Section_filtering}
Let a field in single-photon state $|1_\xi\rangle$ drives a quantum system $G=(S,L,H)$. The reflected field is observed by a detector (see Fig.\ref{FigFiltering}).

\begin{figure}[!htb]
	\begin{center}
		\includegraphics[width=0.4\textwidth]{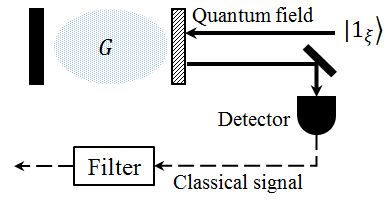}
		\caption{Schematic of output field detection and filtering}
		\label{FigFiltering}
	\end{center}
\end{figure}

When the field is in the vacuum state, the quantum filtering equation has been frequently utilized in researches \cite{RN1343}. However, when the field is in single-photon state, a source model $G_s$ for the field should be considered such that the total model $G_T=G\vartriangleleft G_s$ has a vacuum input \cite{gough2012quantum}. The series product $\vartriangleleft$ has been defined in \cite{RN443}. The source model of a single-photon $|1_\xi\rangle$ with temporal shape $\xi_t$ is \cite{gough2012quantum,JMN2011}:
\begin{equation}\label{Ancilla_SinglePhoton}
G_{1_\xi}=\left(1,\frac{\xi_t}{\sqrt{\int_{t}^{\infty}|\xi_s|^2ds}}\sigma_-,0\right)
\end{equation}
Thus, The filtering equation in the case of a single-photon input for any system operator $X$ could be written done easily by substituting the total system $G_T$ model in equations of filtering for a system driven by vacuum state and tracing over the field as \cite{gough2012quantum}:
\begin{strip}
\begin{eqnarray}
d\pi_t^{11}(X)&=&\{\pi_t^{11}(\mathcal{L}X)+\pi_t^{01}(S^\dagger [X,L])\xi^*_t+\pi_t^{10}([L^\dagger,X]S)\xi_t+\pi_t^{00}(S^\dagger XS-X)|\xi_t|^2\}dt\nonumber\\
&+&\{\pi_t^{11}(XL+L^\dagger X)+\pi_t^{01}(S^\dagger X)\xi^*_t+\pi_t^{10}(XS)\xi_t-\pi_t^{11}(X)K_t\}dW(t)\nonumber\\
d\pi_t^{10}(X)&=&\{\pi_t^{10}(\mathcal{L}X)+\pi_t^{00}(S^\dagger[X,L])\xi^*_t\}dt
+\{\pi_t^{10}(XL+L^\dagger X)+\pi_t^{00}(S^\dagger X)\xi^*_t-\pi_t^{10}(X)K_t\}dW(t)\nonumber\\
d\pi_t^{01}(X)&=&\{\pi_t^{01}(\mathcal{L}X)+\pi_t^{00}([L^\dagger,X]S)\xi_t\}dt
+\{\pi_t^{01}(XL+L^\dagger X)+\pi_t^{00}(XS)\xi_t-\pi_t^{01}(X)K_t\}dW(t)\nonumber\\
d\pi_t^{00}(X)&=&\{\pi_t^{00}(\mathcal{L}X)\}dt
+\{\pi_t^{00}(XL+L^\dagger X)-\pi_t^{00}(X)K_t\}dW(t)
\label{SME_X_HomodyneDet}
\end{eqnarray}
and
\begin{eqnarray}
d\pi_t^{11}(X) &=& \{\pi_t^{11}(\mathcal{L}X)+\pi_t^{01}(S^\dagger [X,L])\xi^*_t+\pi_t^{10}([L^\dagger,X]S)\xi_t+\pi_t^{00}(S^\dagger XS-X)|\xi_t|^2\}dt\nonumber\\
&+&\{\nu_t^{-1}\left(\pi_t^{11}(L^\dagger XL)+\pi_t^{01}(S^\dagger XL)\xi^*_t+\pi_t^{10}(L^\dagger XS)\xi_t+\pi_t^{00}(S^\dagger XS)|\xi_t|^2\right)-\pi_t^{11}(X)\}dN(t)\nonumber\\
d\pi_t^{10}(X) &=& \{\pi_t^{10}(\mathcal{L}X)+\pi_t^{00}(S^\dagger[X,L])\xi^*_t\}dt
+\{\nu_t^{-1}\left(\pi_t^{10}(L^\dagger XL)+\pi_t^{00}(S^\dagger XL)\xi^*_t\right)-\pi_t^{10}(X)\}dN(t)\nonumber\\
d\pi_t^{01}(X) &=& \{\pi_t^{01}(\mathcal{L}X)+\pi_t^{00}([L^\dagger,X]S)\xi_t\}dt
+\{\nu_t^{-1}\left(\pi_t^{01}(L^\dagger XL)+\pi_t^{00}(L^\dagger XS)\xi_t\right)-\pi_t^{01}(X)\}dN(t)\nonumber\\
d\pi_t^{00}(X) &=& \{\pi_t^{00}(\mathcal{L}X)\}dt
+\{\nu_t^{-1}\left(\pi_t^{00}(L^\dagger XL)\right)-\pi_t^{00}(X)\}dN(t)
\label{SME_X_PhotonDet}
\end{eqnarray}
\end{strip}
for Homodyne and photon-counting detector, respectively. Here, $\mathcal{L}X=-\mathrm{i}[X,H]+L^\dagger X L-\frac{1}{2}\left(L^\dagger L X+X L^\dagger L\right)$, $K_t=\pi_t^{11}(L+L^\dagger)+\pi_t^{10}(S^\dagger)\xi^*_t+\pi_t^{01}(S)\xi_t$, and $\nu_t=\pi_t^{11}(L^\dagger L)+\pi_t^{01}(S^\dagger L)\xi^*_t+\pi_t^{10}(L^\dagger S)\xi_t+\pi_t^{00}(I)|\xi_t|^2$. Also, $dW(t)=dY(t)-K_tdt$ and $dN(t)=dY(t)-\nu_tdt$ are the Wiener and Poisson process, respectively, related to measurements $Y(t)$ which adhere $dW^2=dt$, $dN^2=dN$ and $\textbf{E}[dN(t)]=\langle L^\dagger L\rangle dt$.
It worth to note that $\pi_t^{10}(X)=\pi_t^{01}(X)^\dagger$.
Considering $tr\{\rho X(t)\}=tr\{\rho(t)X\}$ where $tr$ is the abbreviation of trace, then, dynamics of the density operator can also be derived which are given in \cite{gough2012quantum}:
for the Homodyne and the photon-counting detector, respectively. Here, $\mathcal{L}^\star\rho_t=-\mathrm{i}[H,\rho]+L\rho L^\dagger-\frac{1}{2}\left(L^\dagger L\rho+\rho L^\dagger L\right)$, $K_t=tr\{(L+L^\dagger)\rho^{11}(t)\}+tr\{S\rho^{01}(t)\}\xi_t+tr\{S^\dagger\rho^{10}(t)\}\xi^*_t$, and $\nu_t=tr\{\rho_t^{11}L^\dagger L\}+tr\{\rho_t^{10}S^\dagger L\}\xi^*_t+tr\{\rho_t^{01}L^\dagger S\}\xi_t+tr\{\rho_t^{00}I\}|\xi_t|^2$.

%%%%%%%%%%%%%%%%%%%%%%%%%%%%%%%%%%%%%%%%%%%%%%%%%%%%%%%%%%%%%%%%%%%%%%%%%%%%%%%%%%%%%%%%%%%%%%%%%%
\section{Cavity evolution driven by single-photon}
\label{Section_filtering_Cavity}
The behavior of a quantum system is studied through observables. In particular, The behavior of a cavity could be investigated by observing the number of photons in the cavity. It is given by tracing over the number operator $n=a^\dagger a$ as $\mathbf{Tr}[\rho(t)n]$ in the Schrodinger picture or $\mathbf{Tr}[\rho(0)n(t)]$ in the Heisenberg picture. So, by considering $X=n$, It is possible to disclose the behavior of a cavity driven by fields in single-photon state conditioned on Homodyne or photon-counting measurements. 
\subsection{Homodyne Detection}
Considering the Homodyne detector, we have:

\begin{theorem}\label{Theorem_HD}
	The number of photons inside a one-sided cavity (\ref{Model_1S_Cavity}), which is driven by field in single-photon state and conditioned on Homodyne measurements, evolves as:
	\begin{strip}
	\begin{eqnarray}\label{C_SME_HD}
	d\pi_t^{11}(n)&=&\{-\kappa\pi_t^{11}(n)-\sqrt{\kappa}\pi_t^{01}(a)\xi^*_t-\sqrt{\kappa}\pi_t^{10}(a^\dagger)\xi_t\}dt+\{\sqrt{\kappa}\pi_t^{11}(na+a^\dagger n)+\pi_t^{01}(n)\xi^*_t+\pi_t^{10}(n)\xi_t-\pi_t^{11}(n)K_t\}dW\nonumber\\
	d\pi_t^{10}(n)&=&\{-\kappa\pi_t^{10}(n)-\sqrt{\kappa}\pi_t^{00}(a)\xi^*_t\}dt+\{\sqrt{\kappa}\pi_t^{10}(na+a^\dagger n)+\pi_t^{00}(n)\xi^*_t-\pi_t^{10}(n)K_t\}dW\nonumber\\
	d\pi_t^{01}(n)&=&\{-\kappa\pi_t^{01}(n)-\sqrt{\kappa}\pi_t^{00}(a^\dagger)\xi_t\}dt+\{\sqrt{\kappa}\pi_t^{01}(na+a^\dagger n)+\pi_t^{00}(n)\xi_t-\pi_t^{01}(n)K_t\}dW\nonumber\\
	d\pi_t^{00}(n)&=&\{-\kappa\pi_t^{00}(n)\}dt+\{\sqrt{\kappa}\pi_t^{00}(na+a^\dagger n)-\pi_t^{00}(n)K_t\}dW\nonumber\\
	\end{eqnarray}
	\end{strip}
where $K_t=\sqrt{\kappa}\pi_t^{11}(a+a^\dagger)+\pi_t^{01}(I)\xi^*_t+\pi_t^{10}(I)\xi_t$.
\end{theorem}
\begin{proof}
	The SME Eqs. \ref{SME_X_HomodyneDet} give the evolution of any system operator $X$. So, we consider $X=n$ and substitute $S$, $L$, and $H$ from Eq. (\ref{Model_1S_Cavity}) into Eq.\ref{SME_X_HomodyneDet}. The first term of Eq. (\ref{SME_X_HomodyneDet}) becomes:
	\begin{eqnarray}
	\pi_t^{ij}(\mathcal{L}n)&=&\pi_t^{ij}(-\mathrm{i}n\Delta a^\dagger a+\mathrm{i}\Delta a^\dagger an+\kappa a^\dagger na-\frac{1}{2}\kappa a^\dagger an\nonumber\\
	&-&\frac{1}{2}\kappa a^\dagger a),\qquad i,j=0,1 \nonumber
	\end{eqnarray}
	Replacing $n=a^\dagger a$ and using the commutation relation $\left[a,a^\dagger\right]=I$ results to
	\begin{eqnarray}
	\pi_t^{ij}(\mathcal{L}n)&=&-\kappa\pi_t^{ij}(n) \nonumber
	\end{eqnarray} 
	For the term $\pi_t^{ij}(S^\dagger [X,L])$ we have:
	\begin{eqnarray}
	\pi_t^{ij}(S^\dagger [X,L])&=&\pi_t^{ij}(\sqrt{\kappa}(na-an))\nonumber\\
	&=&\pi_t^{ij}(\sqrt{\kappa}(a^\dagger aa-aa^\dagger a))\nonumber\\
	&=&\pi_t^{ij}(\sqrt{\kappa}(a^\dagger a-aa^\dagger)a) \nonumber
	\end{eqnarray}
	using commutation relation leads to
	\begin{eqnarray}
	\pi_t^{ij}(S^\dagger[X,L])&=&-\kappa\pi_t^{ij}(a). \nonumber
	\end{eqnarray}
	The other terms could be derived in the same way. Finally, putting these term together leads to Eq. (\ref{C_SME_HD}).$\blacksquare$
\end{proof}
It could be seen that the first and higher order operators in terms of $a$ and $a^\dagger$ as well as $\pi_t{ij}(I)$ are required in Eqs. (\ref{C_SME_HD}). Therefore, dynamical equations for these operators must be derived. The complexity of stochastic equations for higher-order operators, such as $a^\dagger aa$, increase by the order. In simulations, the higher-order terms are usually neglected by some rational assumptions. For instance, assuming the cavity has no photon initially, and the input field is in single-photon state lead to the fact that higher-order terms vanish. The dynamical equations for creation, annihilation, and $I$ are:
\begin{strip}
\begin{eqnarray}
d\pi_t^{11}(a^\dagger)&=&\{-(\mathrm{i}\Delta+\kappa/2)\pi_t^{11}(a^\dagger)-\sqrt{\kappa}\pi_t^{01}(I)\xi^*_t\}dt+\{\sqrt{\kappa} \pi_t^{11} (n+a^\dagger a^\dagger )+\pi_t^{01} (a^\dagger ) \xi^*_t+\pi_t^{10} (a^\dagger )\xi_t-\pi_t^{11} (a^\dagger ) K_t \}dW \nonumber\\
d\pi_t^{10} (a^\dagger )&=&\{-(\mathrm{i}\Delta+\kappa/2) \pi_t^{10} (a^\dagger )-\sqrt{\kappa} \pi_t^{00} (I) \xi^*_t\}dt+\{\sqrt{\kappa} \pi_t^{10} (n+a^\dagger a^\dagger )+\pi_t^{00} (a^\dagger ) \xi^*_t-\pi_t^{10} (a^\dagger ) K_t \}dW\nonumber\\
d\pi_t^{01} (a^\dagger )&=&\{-(\mathrm{i}\Delta+\kappa/2) \pi_t^{01} (a^\dagger )\}dt+\{\sqrt{\kappa} \pi_t^{01} (n+a^\dagger a^\dagger )+\pi_t^{00} (a^\dagger )\xi_t-\pi_t^{01} (a^\dagger ) K_t \}dW\nonumber\\
d\pi_t^{00} (a^\dagger)&=&\{-(\mathrm{i}\Delta+\kappa/2) \pi_t^{00} (a^\dagger )\}dt+\{\sqrt{\kappa} \pi_t^{00} (n+a^\dagger a^\dagger )-\pi_t^{00} (a^\dagger ) K_t \}dW\nonumber\\
d\pi_t^{11}(a)&=&\{-(i\Delta+\kappa/2) \pi_t^{11} (a)-\sqrt{\kappa} \pi_t^{10} (I)\xi_t\}dt+\{\sqrt{\kappa} \pi_t^{11} (n+aa)+\pi_t^{01} (a) \xi^*_t+\pi_t^{10} (a)\xi_t-\pi_t^{11} (a) K_t \}dW\nonumber\\
d\pi_t^{10} (a)   &=&\{-(\mathrm{i}\Delta+\kappa/2) \pi_t^{10} (a)\}dt+\{\sqrt{\kappa} \pi_t^{10} (n+aa)+\pi_t^{00} (a) \xi^*_t-\pi_t^{10} (a) K_t \}dW\nonumber\\
d\pi_t^{01} (a)   &=&\{-(i\Delta+\kappa/2) \pi_t^{01} (a)-\sqrt{\kappa} \pi_t^{00} (I)\xi_t\}dt+\{\sqrt{\kappa} \pi_t^{01} (aa+n)+\pi_t^{00} (a)\xi_t-\pi_t^{01} (a) K_t \}dW \nonumber\\
d\pi_t^{00} (a)   &=&\{-(\mathrm{i}\Delta+\kappa/2) \pi_t^{00} (a)\}dt+\{\sqrt{\kappa} \pi_t^{00} (aa+n)-\pi_t^{00} (a) K_t \}dW\nonumber\\
d\pi_t^{10} (I)    &=& \{\sqrt{\kappa} \pi_t^{10} (a+a^\dagger )+\pi_t^{00} (I) \xi^*_t-\pi_t^{10} (I) K_t \}dW\nonumber\\
d\pi_t^{01} (I)    &=&\{\sqrt{\kappa} \pi_t^{01} (a+a^\dagger )+\pi_t^{00} (I)\xi_t-\pi_t^{01} (I) K_t \}dW\nonumber\\
d\pi_t^{00} (I)    &=&\{\sqrt{\kappa} \pi_t^{00} (a+a^\dagger )-\pi_t^{00} (I) K_t \}dW
\end{eqnarray}
\end{strip}
These equations can easily be written done by substituting $a$, $a^\dagger$, and $I$ into Eq. (\ref{C_SME_HD}) and doing a bit of calculation by applying commutation relation.

\subsection{Photon Detection}
Assume that the cavity is observed by a photon-counting detector. In this situation the filtering equations are:
\begin{theorem}\label{Theorem_PD}
	The number of photons inside a one-sided cavity (\ref{Model_1S_Cavity}), which is driven by field in single-photon state and conditioned on photon-counting measurements, evolves as:
	\begin{strip}
		\begin{eqnarray}\label{C_SME_PD}
		d\pi_t^{11}(n)&=&\{-\kappa\pi_t^{11}(n)-\sqrt{\kappa}\pi_t^{01}(a)\xi^*_t-\sqrt{\kappa}\pi_t^{10}(a^\dagger)\xi_t\}dt\nonumber\\
		&+&\{\nu_t^{-1}\left(\kappa\pi_t^{11}(a^\dagger na)+\sqrt{\kappa}\pi_t^{01}(na)\xi^*_t+\sqrt{\kappa}\pi_t^{10}(a^\dagger n)\xi_t+\pi_t^{00}(n)|\xi_t|^2\right)-\pi_t^{11}(n)\}dN \nonumber\\
		d\pi_t^{10}(n)&=&\{-\kappa\pi_t^{10}(n)-\sqrt{\kappa}\pi_t^{00}(a)\xi^*_t\}dt+\{\nu_t^{-1}\left(\kappa\pi_t^{10}(a^\dagger na)+\sqrt{\kappa}\pi_t^{00}(na)\xi^*_t\right)-\pi_t^{10}(n)\}dN \nonumber\\
		d\pi_t^{01}(n)&=&\{-\kappa\pi_t^{01}(n)-\sqrt{\kappa}\pi_t^{00}(a^\dagger)\xi_t\}dt+\{\nu_t^{-1}[\kappa\pi_t^{01}(a^\dagger na)+\sqrt{\kappa}\pi_t^{00}(a^\dagger n)\xi_t]-\pi_t^{01}(n)\}dN\nonumber\\
		d\pi_t^{00}(n)&=&\{-\kappa\pi_t^{00}(n)\}dt+\{\nu_t^{-1}\left(\kappa\pi_t^{00}(a^\dagger na)\right)-\pi_t^{00}(n)\}dN 
		\end{eqnarray}
	\end{strip}
	where $\nu_t=\kappa\pi_t^{11}(n)+\sqrt{\kappa}\pi_t^{01}(a)\xi^*_t+\sqrt{\kappa}\pi_t^{10}(a^\dagger)\xi_t+\pi_t^{00}(I) |\xi_t|^2$.
\end{theorem}
\begin{proof}
	This theorem could be proved in the same way as theorem \ref{Theorem_HD} by considering $X=n$ and substituting $S$, $L$, and $H$ from Eq. (\ref{Model_1S_Cavity}) into Eq. (\ref{SME_X_PhotonDet}).
	$\blacksquare$
\end{proof}
Similar to the case of Homodyne detection, the evolution of the creation, annihilation, and unit operators could be derived as:
\begin{strip}
	\begin{eqnarray}
	d\pi_t^{10}(a^\dagger)&=&\{-(\mathrm{i}\Delta+\kappa/2)\pi_t^{10}(a^\dagger)-\sqrt{\kappa}\pi_t^{00}(I)\xi^*_t\}dt+\{\nu_t^{-1}\left(\kappa\pi_t^{10}(a^\dagger a^\dagger a)+\sqrt{\kappa}\pi_t^{00}(n)\xi^*_t\right)-\pi_t^{10}(a^\dagger)\}dN \nonumber\\
	d\pi_t^{00}(a^\dagger)&=&\{-(\mathrm{i}\Delta+\kappa/2)\pi_t^{00}(a^\dagger)\}dt+\{\nu_t^{-1}\left(\kappa\pi_t^{00}(a^\dagger a^\dagger a)\right)-\pi_t^{00}(a^\dagger)\}dN \nonumber\\
	d\pi_t^{01}(a)&=&\{-(\mathrm{i}\Delta+\kappa/2)\pi_t^{01}(a)-\sqrt{\kappa}\pi_t^{00}(I)\xi_t\}dt+\{\nu_t^{-1}[\kappa\pi_t^{01}(a^\dagger aa)+\sqrt{\kappa}\pi_t^{00}(n)\xi_t]-\pi_t^{01}(a)\}dN \nonumber\\
	d\pi_t^{00}(a)&=&\{-(\mathrm{i}\Delta+\kappa/2)\pi_t^{00}(a)\}dt+\{\nu_t^{-1}\left(\kappa\pi_t^{00}(a^\dagger aa)\right)-\pi_t^{00}(a)\}dN \nonumber\\
	d\pi_t^{00}(I)&=&\{\nu_t^{-1}\left(\kappa\pi_t^{00}(n)\right)-\pi_t^{00}(I)\}dN  
	\end{eqnarray}
\end{strip} 
$dN$ is a Poisson process that may have different realizations.

%%%%%%%%%%%%%%%%%%%%%%%%%%%%%%%%%%%%%%%%%%%%%%%%%%%%%%%%%%%%%%%%%%%%%
\section{Results and Discussions}
\label{Section_Results}
In order to simulate the behavior of a cavity driven by single-photon, the following assumptions are considered:
\begin{itemize}
	\item The cavity has no photon initially. So, the higher-order operators are vanishes.
	\item We have $\pi_t^{01}(X)=\pi_t^{10}(X^\dagger)^\dagger$. Thus, $\pi_t^{10}(a^\dagger)=\pi_t^{01}(a)^*$, $\pi_t^{01}(a^\dagger)=\pi_t^{10}(a)^*$, and $\pi_t^{10}(I)=\pi_t^{01}(I)^*$.
	\item  The initial conditions for any operator $X$ are given by $\pi_t^{01}(X)=\pi_t^{10}(X)=0$ and $\pi_t^{11}(X)=\pi_t^{00}(X)=\langle\eta|X|\eta\rangle$. As a result, $\pi_t^{01}(a)=\pi_t^{10}(a)=\pi_t^{01}(a^\dagger)=\pi_t^{10}(a^\dagger)=\pi_t^{01}(n)=\pi_t^{10}(n)=\pi_t^{01}(I)=\pi_t^{10}(I)=0$, $\pi_t^{11}(a)=\pi_t^{00}(a)=\pi_t^{11}(a^\dagger)=\pi_t^{00}(a^\dagger)=\pi_t^{11}(n)=\pi_t^{00}(n)=0$, and $\pi_t^{11}(I)=\pi_t^{00}(I)=1$.
	\item As $\pi_t^{00}(n)$, $\pi_t^{00}(a)$, and $\pi_t^{00}(a^\dagger)$ are equal to zero at $t=0$, then $d\pi_t^{00}(n)=0$, $d\pi_t^{00}(a)=0$, $d\pi_t^{10}(n)=0$, $d\pi_t^{01}(n)=0$, and $d\pi_t^{00}(a^\dagger)=0$. 
\end{itemize}
Another point in regards of doing simulation is that presented stochastic dynamical equations are in Ito representation while the Stratonovich definition is easier to simulate numerically due to normal rules of calculus do not apply to the Ito integral. Also, SME is a stochastic process that multiple realizations (or trajectories) could be simulated.

Let a single-photon has the shape:
\begin{eqnarray}\label{SinglePhoton}
\xi_t&=&\sqrt{\gamma}e^{-\frac{\gamma}{2}(t-t_c)}u(t-t_0)
\end{eqnarray}
where $\gamma$ is the decay rate of a two-level atom which emits this photon and $u(t-t_0)$ is the Heaviside step function \cite{RN1157}. One hundred different trajectories for the number of photons inside a cavity driven by a single-photon Eq.(\ref{SinglePhoton}) and observed by Homodyne detector are depicted in Fig.\ref{HD_Trajectoreis} for $t_0=3$, $\kappa=0.1$, and $\Delta=0$.
An acceptable matching between the average of trajectories and the ME can be seen. Also, the cavity may excite to the $|1\rangle$ in some realizations. In addition, the cavity with the larger decay rate, the faster emits the photons to the bath. This phenomenon has been shown in Fig.\ref{HD_Diff_kappa} by the ME dynamics.

\begin{figure}[!htb]
	\begin{center}
		\includegraphics[width=0.45\textwidth]{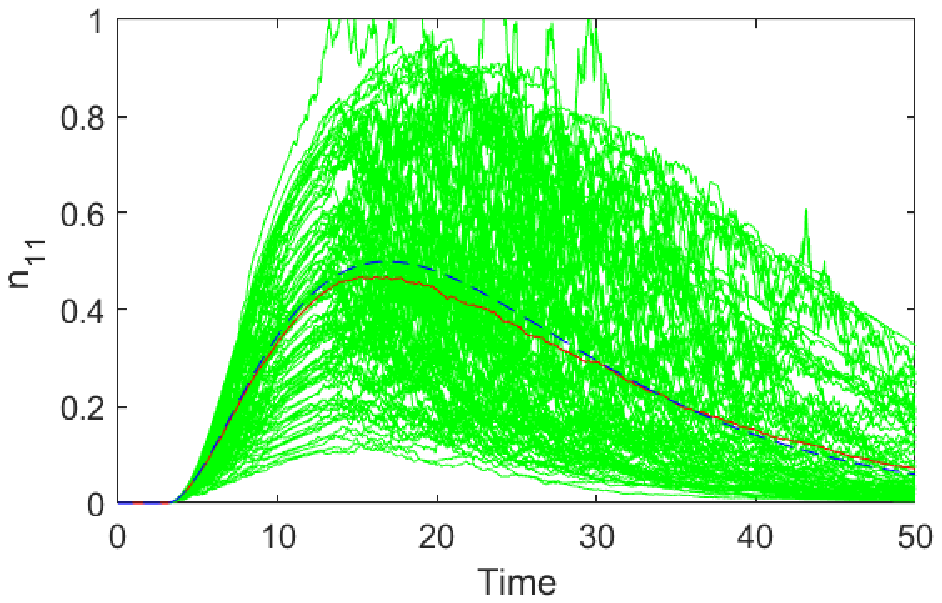}
		\caption{The number of photons inside a cavity driven by single-photon with $t_0=3$, $\kappa=0.1$, and $\Delta=0$. Gray(Green): different trajectories; Black solid line (red): the average of different trajectories; Dotted line (blue): Master equation dynamic}
		\label{HD_Trajectoreis}
	\end{center}
\end{figure}
\begin{figure}[!htb]
	\begin{center}
		\includegraphics[width=0.45\textwidth]{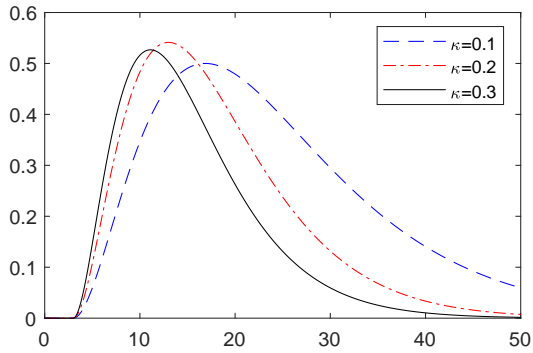}
		\caption{The master equation for different $\kappa$}
		\label{HD_Diff_kappa}
	\end{center}
\end{figure}

When a cavity is driven by a single-photon and is observed by a photon-counting detector, the estimated state of the cavity gradually goes from $|0\rangle$ to $|1\rangle$. But, the estimated state will collapse to $|0\rangle$ when a photon is observed by the detector. One hundred different trajectories for this situation are displayed in Fig. \ref{PD_Trajectoreis}. A satisfactory match between average of different trajectories and ME can be seen.
\begin{figure}[!htb]
	\begin{center}
		\includegraphics[width=0.45\textwidth]{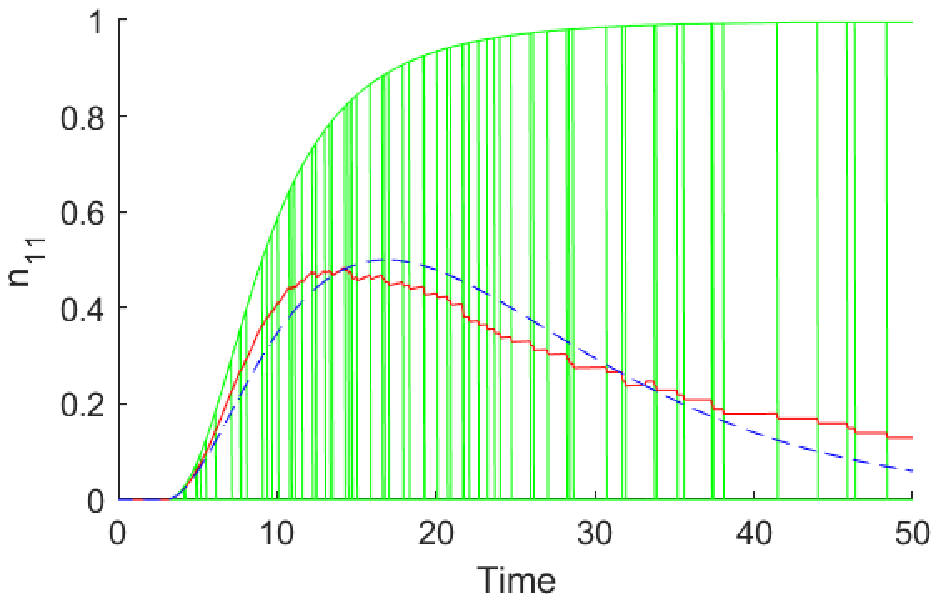}
		\caption{The number of photons inside a cavity driven by single-photon with $t_0=3$, $\kappa=0.1$, and $\Delta=0$. Gray(Green): different trajectories; Black solid line (red): the average of different trajectories; Dotted line (blue): Master equation dynamic}
		\label{PD_Trajectoreis}
	\end{center}
\end{figure}
 
%%%%%%%%%%%%%%%%%%%%%%%%%%%%%%%%%%%%%%%%%%%%%%%%%%%%%%%%%%%%%%%%%%%%%
\section{Conclusion}
\label{Section_Conclusion}
This paper investigated the behavior of a quantum cavity driven by filed in single-photon state. The number operator was selected to reveal the number of photons inside the cavity in time. The filtering equations were derived for the number operator in two situations: the output field is observed by either Homodyne or photon-counting detector. Also, some assumptions were made to simplify the problem. Finally, derived stochastic master equations were simulated, and the results were compared with the conventional master equation.

%\small
%\begin{thebibliography}
%\bibliographystyle{plain}
\bibliographystyle{IEEEtran}
\bibliography{Citations}{}
%\end{thebibliography}

\end{document}